\documentclass[journal,12pt,onecolumn]{IEEEtran}

\usepackage{amsmath,amsfonts}
\usepackage{amssymb,amsthm}
\usepackage{makecell}
\usepackage{tikz}
\usepackage{caption}
	
\newtheorem{theorem}{Theorem}[section]
\newtheorem{lemma}{Lemma}[section]
\newtheorem{definition}{Definition}[section]

\newtheorem{proposition}{Proposition}[section]
\newtheorem{example}{Example}[section]
\newtheorem{remark}{Remark}[section]

\begin{document}
\title{Constructing Sum-Rank Metric Codes from Quadratic Galois Extensions of Function Fields}
\author{Yunlong Zhu\thanks{Y. Zhu is with the Department of Mathematics, School of Mathematics, Sun Yat-sen University, Guangzhou 510275, China (e-mail: zhuylong3@mail2.sysu.edu.cn).},
Chang-An Zhao$^*$\thanks{
C.-A. Zhao is with the School of Mathematics, Sun Yat-sen University, Guangzhou 510275, China, and also with the Guangdong Key Laboratory of Information Security Technology, Guangzhou 510006, China (e-mail: zhaochan3@mail.sysu.edu.cn). \\$*${Corresponding author.}}

}
\date{\today}
\maketitle

\begin{abstract}
This paper introduces new constructions of sum-rank metric codes derived from algebraic function fields, as existing results on such codes remain limited. A major challenge lies in the determination of their parameters. We address this issue by employing quadratic Galois extensions, proposing two general constructions of $2\times2$ sum-rank codes. Analogous to algebraic geometry codes in the Hamming metric, our codes achieve a larger block length compared to existing constructions. We determine explicit parameters including dimensions and minimum distances of our codes, and we present an illustrative example using elliptic function fields. Finally, we discuss the asymptotic behavior of our codes and compare them with the Gilbert-Varshamov-like bound for sum-rank metric codes.

\end{abstract}

\section{Introduction}
Algebraic function fields over finite fields provide a fundamental framework for modern error-correcting codes in the Hamming metric \cite{Goppa}, as the function field perspective enables a deeper analysis of code parameters, including minimum distance and duality properties \cite{Stich}. Specifically, a linear $[n,k]$ Hamming metric code $C$ over $\mathbb{F}_q$ is a $k$-dimensional subspace of $\mathbb{F}_q^n$ with minimum Hamming distance
\[
	d_H:=\min_{c_1\neq c_2}\{wt_H(c_1-c_2)\mid c_1,c_2\in C\},
\]
where $wt_H$ denotes the Hamming weight, counting non-zero coordinate positions. Classical algebraic geometry (AG) codes are constructed by evaluating functions from Riemann-Roch spaces at rational places, enabling explicit determination of code parameters.

For a linear $[n,k,d_H]$ code $C$, the Singleton bound states that
\[
	d_H\leq n-k+1.
\]
When equality holds, the code is called a maximum distance separable (MDS) code. Reed-Solomon (RS) codes \cite{Huffman}, derived from rational function fields $\mathbb{F}_q(x)/\mathbb{F}_q$ with $x$ transcendental over $\mathbb{F}_q$, represent prominent examples of Hamming metric MDS codes. Their maximum length of $q+1$ over $\mathbb{F}_q$ aligns with the main conjecture of MDS codes.

Recently, sum-rank metric codes have garnered significant attention due to their applications in multishot network coding \cite{Liu-multishot,Penas-multishot,Napp-multishot,Nobrega-multishot}, space-time coding \cite{Shehadeh}, and distributed storage \cite{Penas-skew,Penas-mr}. For further details and applications, we refer the reader to \cite{Penas-book}. The basic definitions are provided as follows.
\begin{definition}
Given $s\ge1$ and positive integers $n_1,\ldots,n_s,m_1,\ldots,m_s$ with $n_i\le m_i$ and $m_1\ge m_2\cdots\ge m_s$, define
\[
	\Pi :=\Pi_q(n_1\times m_1|\cdots|n_s\times m_s)=\bigoplus_{i=1}^s\mathbb{F}_q^{n_i\times m_i}
\]
where $\mathbb{F}_q^{n_i\times m_i}$ denotes the space of $n_i\times m_i$ matrices over $\mathbb{F}_q$. Note that $\Pi$ is an $\mathbb{F}_q$-linear space of dimension $\sum_{i=1}^sn_im_i$. The \textbf{sum-rank weight} of
\[
	X:=(X_1,\ldots,X_s)\in\Pi
\]
is defined as
\[
	wt_{sr}(X) :=\sum\limits_{i=1}^s{\rm rank}(X_i).
\]
The \textbf{sum-rank distance} between $X$ and $Y$ in $\Pi$ is
\[
	d_{sr}(X,Y)=wt_{sr}(X-Y).
\]
\end{definition}
The sum-rank weight and distance induce a metric on $\Pi$. We present the definition as follows.
\begin{definition}
A linear \textbf{sum-rank metric code} $C_{sr}$ is an $\mathbb{F}_q$-linear subspace of $\Pi$. The \textbf{minimum sum-rank distance} of $C_{sr}$ is defined as
\[
	d_{sr}(C):=\min_{X\neq Y}\{wt_{sr}(X-Y)|X,Y\in C_{sr}\}.
\]
\end{definition}
Sum-rank metric codes generalize Hamming metric codes of length $s$ when $n_i=m_i=1$ for all $1\le i\le s$. When $s=1$, they reduce to \textbf{rank metric} codes, which were first introduced by Delsarte \cite{Delsarte}. 

Fundamental properties and bounds for sum-rank metric codes are examined in \cite{Abiad,Abida-bound,Byrne,Ott-bound}. The Singleton bound was established in \cite{Byrne} and \cite{Penas-mr}. In the general case, let $j$ and $\delta$ be the unique integers satisfying $d-1=\sum_{i=1}^{j-1}n_i+\delta$ with $0\le\delta\le n_j-1$. Then
\[
	|C|\le q^{\sum_{i=j}^sm_in_i-m_j\delta}.
\]
A code $C\subseteq\Pi$ that attains this bound is called a \textbf{maximum sum-rank distance (MSRD)} code. For the special case where $n_1=\cdots=n_s=n$ and $m_1=\cdots=m_s=m$, this bound takes the form
\[
	|C|\le q^{m(sn-d_{sr}+1)}.
\]
Throughout this paper, we concentrate on sum-rank metric codes with these uniform parameters.

Let $N:=sn$. For fixed positive real numbers $\mathcal{R}_{sr},\delta_{sr}\in(0,1)$ satisfying
\begin{align*}
\mathcal{R}_{sr}<&\delta_{sr}^2-\delta_{sr}(2+\frac{2}{N})+1+\frac{2}{N}+\frac{1}{N^2}\\
				&-\frac{\sum_{i=1}^{d_{sr}N-1}\log_q(1+\frac{s-1}{i})+\log_q(\delta_{sr}N-1)}{Nm}-\frac{\log_q(\gamma_q)}{nm},
\end{align*}
where $\gamma_q=\prod_{i=1}^{\infty}(1-q^{-i})^{-1}$, there exists a sum-rank code with rate at least $\mathcal{R}_{sr}$ and relative minimum sum-rank distance at least $\delta_{sr}$. When $m=\xi n$ goes to the infinity and $m\in\omega(\log_q(s))$, where $\xi$ is a constant, we have
\[
	\mathcal{R}_{sr}\sim \delta_{sr}^2-\delta_{sr}(1+\frac{1}{\xi})+1.
\]
Recent work \cite{Ott-bound} demonstrates that random linear sum-rank metric codes attain the asymptotic Gilbert-Varshamov-like (GV-like) bound with high probability. Consider the case where $n=m$ and the number of blocks $s$ goes to infinity. Then the previous bound is asymptotically provided by \cite{Berardini}:
\begin{align*}
\mathcal{R}_{sr}<&(\delta_{sr}-1)^2-\frac{\delta_{sr}}{n}\log_q(1+\frac{1}{\delta_{sr}n})-\frac{\log_q(1+\delta_{sr}n)}{n^2}\\
				&-\frac{\log_q(\gamma_q)}{n^2}+o(1).
\end{align*}

Constructions of ``good" sum-rank metric codes are available in several works \cite{Abiad,Byrne-anticode,Moreno,Chen-explicit,Penas-skew,Penas-bch,Penas-pmds,Penas-dtextend,Penas-msrd, Neri-twisted,Neri-oneweight}. A standard technique involves constructing $C$ as a classical Hamming metric code of length $sn$ over $\mathbb{F}_{q^m}$ and then applies rank metric code constructions to each $n$-length block. This ensures that the dimension $\kappa$ of $C$ over $\mathbb{F}_{q^m}$ satisfies $d_{sr}+\kappa\le sn+1$, matching the Singleton bound for the Hamming metric. Moreover, algorithms over $\mathbb{F}_{q^m}$ can be employed for decoding sum-rank metric codes, see \cite{Bartz,Chen-decoding,Horman,Puchinger}.

One of the most prominent rank metric codes is the Gabidulin code, which is constructed using linearized polynomials over $\mathbb{F}_{q^m}$, see \cite{Gabidulin,Guo,Islam,Li,Neri-Gabidulin}. Among sum-rank metric codes, linearized Reed-Solomon codes \cite{Penas-skew} represent the most extensively investigated family to date, see \cite{Caruso,Horman-skew,Jerkovits,Neri-twisted,Punchinger-bound}. Gabidulin codes and linearized RS codes are \textbf{maximum rank distance (MRD)} and MSRD respectively, serving as generalizations of RS codes.

\subsection*{Motivations}
While non-trivial algebraic geometric constructions for sum-rank metric codes exist, see \cite{Chen-explicit,Penas-pmds}, most results are constructed using towers of function fields as illustrated below. In \cite{Berardini}, E. Berardini and X. Caruso established the first geometric construction of sum-rank metric codes. Their approach employs the Ore polynomial ring \cite{Ore} over an extension field derived from the completion of $F$ at specified places, and subsequently develops the theory of Riemann-Roch spaces over this ring. This framework allows them to define linearized algebraic geometry codes and present explicit constructions of sum-rank metric codes with parameters consistent with classical algebraic geometry codes.

A principal motivation for constructing codes via algebraic function fields lies in the availability of numerous evaluation places bounded by the Hasse-Weil bound for higher genus. Consequently, codes with long lengths can be constructed over small finite fields. Classical algebraic geometry codes are known to asymptotically surpass the Gilbert-Varshamov (GV) bound \cite{TVZ-bound}. This naturally raises the question: whether sum-rank codes derived from algebraic geometry also possess parameters surpassing the GV bound.

\begin{figure}\label{gc}
	\centering
	\tikzstyle{process} = [rectangle, minimum width=1cm, minimum height=0.5cm, text centered]
	\begin{tikzpicture}[node distance=2.5cm]
		\node[process] (curve) {$F/\mathbb{F}_{q^m}$};
		\node[process, below of=curve] (constant extension) {$\mathbb{F}_{q^m}(x)/\mathbb{F}_{q^m}$};
		\node[process, below of=constant extension] (base field) {$\mathbb{F}_{q}(x)/\mathbb{F}_{q}$};
		\node[process, right of =curve, xshift=2cm] (Places) {$P_1,\ldots,P_n,\ldots,P_{sn},P_i\in F/\mathbb{F}_{q^m}$};
		\node[process, below of =Places] (vectors) {${\bf c_1,\ldots,c_s,c_i}\in\mathbb{F}_{q^m}^n$};
		\node[process, below of =vectors] (matrices) {$X_1,\ldots,X_s,X_i\in\Pi$};
		
		\draw[->] (curve) --   (constant extension);
		\draw[->] (constant extension) --  (base field);
		\draw[->] (Places) -- node[right]{Hamming metric}   (vectors);
		\draw[->] (vectors) --node[right]{rank metric}  (matrices);
	\end{tikzpicture}
	\caption{Tower of function field extensions.}
\end{figure}
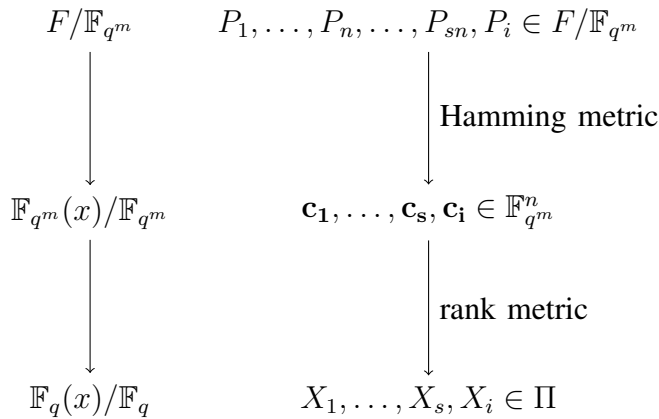

\subsection*{Our contributions}	
In contrast to the construction in \cite{Berardini}, we propose an approach for constructing sum-rank metric codes using function fields. The central element of our construction is the {\it local integral basis} at a place in function field extensions. Specifically, let $\mathcal{O}_P$ be the valuation ring at a place $P$ of $F/\mathbb{F}_q$, and let $F'/F$ be a finite separable function field extension. The integral closure $\mathcal{O}'_P$ of $\mathcal{O}_P$ in $F'$ constitutes a vector space over $\mathcal{O}_P$. Consequently, a function in $F'/\mathbb{F}'_q$ that is regular at the places lying over $P$ corresponds to a matrix over $F$, where $\mathbb{F}'_q$ is a finite extension of $\mathbb{F}_q$. When $F'/F$ is Galois, we define a map on $\mathcal{O}'_P$, via the direct sum of Riemann-Roch spaces. This yields sum-rank metric codes through the evaluation of the function matrix at rational places. Although our construction shares some similarities with \cite{Berardini}, it significantly differs by focusing on the intrinsic structure of function field extensions. By decomposing the Riemann-Roch spaces as a direct sum of subspaces, we explicitly determine the parameters of our sum-rank codes. Furthermore, when $q$ is a special square, we analyze the asymptotic behavior of our codes and compare them with the GV-like bound.

\subsection*{Organization}
The remainder of this paper is organized as follows. Section II introduces the fundamental concepts required for our purpose. In Section III, we develop a general construction of sum-rank metric codes via function field extensions. Section IV provides the explicit sum-rank metric codes using quadratic Galois extensions and presents a new specialized construction, including an instance involving elliptic function fields and the comparison with the GV-like bound. Section V concludes this paper.

\section{Preliminaries and basic concepts}
This section reviews fundamental definitions and notations, primarily based on \cite{Stich}.
\subsection{Places and divisors}
Let $q$ be a prime power. An algebraic function field $F/\mathbb{F}_q$ is a finite extension of $\mathbb{F}_q(x)$ with genus $g(F)$, where $x$ is transcendental over $\mathbb{F}_q$. A place $P$ is the maximal ideal of a valuation ring $\mathcal{O}_P$ in $F/\mathbb{F}_q$. Denote by $\hat{F}_P:=\mathcal{O}_P/P$ the residue field of $P$, which is a finite extension of $\mathbb{F}_q$. The degree of $P$ is defined as
\[
	\deg(P)=[\hat{F}_P:\mathbb{F}_q].
\]
If $\deg(P)=1$, then $P$ is called rational. The prime element (or uniformizing parameter) of $P$ is a function $t\in P$ such that $P=t\mathcal{O}_P$. Any non-zero element $z\in F$ admits a unique representation $z=t^nu$ with $u\in\mathcal{O}^*_P$. The discrete valuation associated with $P$ is 
\[
	v_P:F\to\mathbb{Z}\cup\{\infty\}
\]
, satisfying $v_P(z)=n$ and $v_P(0)=\infty$. The strict triangle inequality is characterized as follows:
\begin{lemma}
For any discrete valuation $v_P$ with associated place $P$,
\[
	v_P(x+y)\ge\min\{v_P(x),v_P(y)\}
\]	
holds for all $x,y\in F$. Moreover, if $v_P(x)\neq v_P(y)$, then $v_P(x+y)=\min\{v_P(x),v_P(y)\}$.
\end{lemma}

A divisor $G$ is defined as a formal sum of places
\[
	G:=\sum\limits_Pn_{P}P
\]
with $n_{P}\in\mathbb{Z}$ and $\deg(G)=\sum\limits_Pn_{P}\deg(P)$. The support of $G$ is defined as
\[
	\text{supp}(G):=\{P\mid n_{P}\neq0\}. 
\]
For any function $z\in F$,
\[
	(z):=\sum\limits_Pv_P(z)P
\]
is called the principal divisor of $z$. Given two divisors $G_1$ and $G_2$, we write $G_1\le G_2$ if $n_P(G_1)\le n_P(G_2)$ for all places $P$. The Riemann-Roch space associated with $G$ is defined as
\[
	\mathcal{L}(G):=\{z\in F\mid(z)+G\ge0\}\cup\{0\}.
\]
This space is a vector space over $\mathbb{F}_q$ with dimension given by the Riemann-Roch theorem:
\begin{lemma}\cite[Theorem 1.4.17]{Stich}
Let $\ell(G)$ denote the dimension of $\mathcal{L}(G)$ over $\mathbb{F}_q$. Then
\[
	\ell(G)\ge\deg(G)-g(F)+1.
\]
Moreover, if $\deg(G)\ge 2g(F)-1$, then
\[
	\ell(G)=\deg(G)-g(F)+1.
\]
\end{lemma}

\subsection{Extensions of function fields}
An algebraic function field $F'/\mathbb{F}_{q^m}$ is called an algebraic extension of $F/\mathbb{F}_q$ if $F'/F$ is an algebraic field extension with $[F':F]<\infty$ and $m\ge1$. A place $P'$ of $F'$ is said lying over a place $P$ of $F$ if $P\subseteq P'$, denoted by $P'|P$. The \textbf{ramification index} $e(P'|P)$ is defined as the integer satisfying
\[
	v_{P'}(x)=e(P'|P)\cdot v_P(x)
\] for all $x\in F$, and the relative degree of $P'$ over $P$ is given by
\[
	f(P'|P)=[\mathcal{O}_{P'}/P':\mathcal{O}_{P}/P].
\]
For any place $P$ of $F$, there exist only finitely many places $P'$ of $F'$ such that $P'|P$. The conorm of $P$ is defined as
\[
	\text{Con}_{F'/F}(P):=\sum\limits_{P'|P}e(P'|P)\cdot P'
\]
and
\[
	\text{Con}_{F'/F}\left(\sum\limits_Pn_{P}P\right):=\sum n_P\cdot \text{Con}_{F'/F}(P).
\]
Consequently, for any non-zero function $z\in F$, the principal divisor of $z$ in $F'$ satisfies
\[
	(z)_{F'}=\text{Con}_{F'/F}((z)_F).
\]
The fundamental equality of places is stated as follows.
\begin{proposition}\cite[Theorem 3.3.11]{Stich}
Let $P$ be a place of $F$, and let $P_1,\ldots,P_n$ be all places of $F'$ lying over $P$. Denote $e_i=e(P_i|P)$ and $f_i=f(P_i|P)$ for $i=1,\ldots,n$. Then
\[
	\sum\limits_{i=1}^ne_if_i=[F':F].
\]
\end{proposition}
If $[F':F]=n$ and there exist exactly $n$ distinct places lying over $P$, then $P$ \textbf{splits completely} in $F'/F$. In this case, $e_i=f_i=1$ for all $i$. If $e(P'|P)>1$ for some $P'$, then $P$ is ramified in $F'/F$. Moreover, if there exists a place $P'$ with $e(P'|P)=n$, then $P$ is \textbf{totally ramified} in $F'/F$.

\subsection{Galois Extensions}
An extension $F'/\mathbb{F}_{q^m}$ of a function field $F/\mathbb{F}_q$ is Galois if $F'/F$ is a finite Galois extension. Specifically, the automorphism group
\[
	\text{Aut}(F'/F)=\{\sigma:F'\to F'\mid\text{$\sigma$ is an isomorphism with $\sigma(a)=a$ for all $a\in F$}\}
\]
has order $[F':F]$. This group is called the $\textbf{Galois group}$ and write $\text{Gal}(F'/F):=\text{Aut}(F'/F)$.

For any place $P$ of $F/\mathbb{F}_q$, the Galois group $\text{Gal}(F'/F)$ acts transitively on the set of places
\[
	\{P'\mid\text{$P'$ lies over $P$}\}
\]
via $\sigma(P')=\{\sigma(x)\mid x\in P'\}$. For any place $P'$ of $F'/\mathbb{F}_{q^m}$, the valuation corresponding to the place $\sigma(P')$ is given by
\[
	v_{\sigma(P')}(y)=v_{P'}(\sigma^{-1}(y)),\quad y\in F'.
\]

\section{A General Setting of Sum-Rank Metric Codes}
In this section, we introduce a general construction of sum-rank metric codes using function field extensions, which differs from the approach proposed by E. Berardini and X. Caruso in \cite{Berardini}.
\subsection{The Main Idea of Matrices over $\mathbb{F}_q$}
Let $F/\mathbb{F}_q$ be a function field, and $F'/\mathbb{F}_{q^m}$ be a Galois extension of $F/\mathbb{F}_q$ with $[F':F]=n$. Suppose that $P$ is a place of $F$, and $t\in F$ is a prime element of $P$. The integral closure $\mathcal{O}'_P$ of $\mathcal{O}_P$ in $F'$ is given by
\[
	\mathcal{O}'_P = \bigcap\limits_{i=1}^n\mathcal{O}_{P_i}.
\]
An integral basis $\{z_1,\ldots,z_n\}$ of $\mathcal{O}'_P$ over $\mathcal{O}_P$ is a basis of $F'/F$ satisfying
\[
	\mathcal{O}'_P=\sum\limits_{i=1}^n\mathcal{O}_Pz_i.
\]
The quotient $V:=\mathcal{O}'_P/t\mathcal{O}'_P$ forms a vector space over the finite field $\hat{F}_P$ with basis
\[
	\{z_1+t\mathcal{O}'_P,\ldots,z_n+t\mathcal{O}'_P\}.
\]
For any $\sigma\in\text{Gal}(F'/F)$ and a place $P'$ with $P'|P$, observe that for any $z\in\mathcal{O}'_P$,
\[
	v_{P'}(\sigma(z))=v_{\sigma^{-1}(P')}(z)\ge0
\]
since $\sigma^{-1}(P')|P$. This implies $\sigma(z)\in\mathcal{O}'_P$, and hence $\sigma(V)=V$. The set
\[
	\{\sigma(z_1)+t\mathcal{O}'_P,\ldots,\sigma(z_n)+t\mathcal{O}'_P\}
\]
forms another basis of $V$ over $\hat{F}_P$. Then $\sigma$ induces an $n\times n$ matrix $\bar{A}_{\sigma}$ over $\hat{F}_P$ with respect to the basis $\{z_i\}$.

Similarly, for any $f\in\mathcal{O}'_P$, consider the multiplication map $\mu_f:\mathcal{O}'_P\to\mathcal{O}'_P$ defined by $\mu_f(z)=fz$. This induces an $\hat{F}_P$-linear map $\bar{\mu}_f:V\to V$ via
\[
	\bar{\mu}_f(z+t\mathcal{O}'_P):=fz+t\mathcal{O}'_P.
\]
This also induces an $n\times n$ matrix $\bar{B}_f$ with entries in $\hat{F}_P$.

\begin{remark}
The matrix $A=(a_{ij})$ associated with $\sigma$ has entries in $\mathcal{O}_P$ due to the integral basis property. Under the reduction map
\[
	\pi:\mathcal{O}_P\to \hat{F}_P,
\]
we obtain $\bar{A}=(\pi(a_{ij}))$. This justifies the notation $\bar{A}_{\sigma}$ and $\bar{B}_f$.
\end{remark}

When $P$ is rational (i.e., $\hat{F}_P=\mathbb{F}_q$), both $\sigma\in\text{Gal}(F'/F)$ and $f\in\mathcal{O}'_P$ induce matrices over $\mathbb{F}_q$. Let $\text{Gal}(F'/F)=\{\sigma_1,\ldots,\sigma_n\}$, and consider functions $f_1,\ldots,f_n\in\mathcal{O}'_P$. We define a linear operator
\begin{align*}
L : \mathcal{O}'_P&\to \mathcal{O}'_P\\
	 g &\mapsto\sum\limits_{i=1}^nf_i\sigma_i(g).
\end{align*}
For any two functions $g_1,g_2\in F$, we have
\[
	L(g_1+g_2)=\sum\limits_{i=1}^nf_i\sigma_i(g_1+g_2)=\sum\limits_{i=1}^nf_i(g_1+g_2)=L(g_1)+L(g_2).
\]
Thus, this $F'$-linear operator naturally induces a linear morphism on $F$. We further define a map $\varepsilon_P$ such that
\[
	\varepsilon_P(L)=\sum\limits_{i=1}^nA_iB_i=(M_{ij})
\]
where $B_i$ and $A_i$ are matrices representing $\mu_{f_i}$ and $\sigma_i$ for $1\le i\le n$, respectively. Then $\varepsilon_P(L)$ is an $n\times n$ matrix over $\mathcal{O}_P$. Moreover, let
\[
	\bar{\varepsilon}_P(L)=\sum\limits_{i=1}^n\bar{A_i}\bar{B_i}=(\pi(M_{ij}))
\]
be the matrix over $\mathbb{F}_q$ induced by $\varepsilon_P(L)$. We now state the following lemma.
\begin{lemma}\label{detmatrix}
\begin{itemize}
\item[(i)] Let $\{z_1,\ldots,z_n\}$ be an integral basis of $\mathcal{O}'_P$ over $\mathcal{O}_P$. Then the matrix representing $L$ with respect to this basis is $\varepsilon_P(L)$.
\item[(ii)] The determinant $\det(\varepsilon_P(L))$ belongs to $\mathcal{O}_P$.
\item[(iii)] The kernel ${\rm ker}(L)\neq\{0\}$ if and only if $\det(\varepsilon_P(L))=0$.
\end{itemize}
\end{lemma}
\begin{proof}
(i) Let $g=g_1z_1+\cdots+g_nz_n\in\mathcal{O}'_P$. For any $\sigma\in\text{Gal}(F'/F)$ and $f\in\mathcal{O}_P$, let $A_{\sigma}$ and $B_f$ denote the matrices defined above. Then
\[
	\mu_f(g)=[g_1,\ldots,g_n]B_f[z_1,\ldots,z_n]^T
\]
and
\[
	\sigma(g)=[g_1,\ldots,g_n]A_{\sigma}[z_1,\ldots,z_n]^T.
\]
Thus $f\sigma(g)=[g_1,\ldots,g_n]B_fA_{\sigma}[z_1,\ldots,z_n]^T$. The result follows from the linearity of the maps $\mu_f$ and $\sigma$.\\
(ii) The result immediately follows since all coefficients of $A_{\sigma_i}$ and $B_{f_i}$ are contained in $\mathcal{O}_P$.\\
(iii) Consider the equation
\[
	\varepsilon_P(L)^T[x_1,\ldots,x_n]^T=[0,\ldots,0].
\]
Any solution $[g_1,\ldots,g_n]$ satisfies $L(g_1z_1+\cdots+g_nz_n)=0$. Therefore ${\rm ker}(L)$ contains nonzero elements if and only if the system has nontrivial solutions, which occurs precisely when $\det(\varepsilon_P(L))=0$.
\end{proof}
\subsection{Sum-rank Codes in the General Case}
We now present the construction of sum-rank metric codes. Let $Q_1,\ldots,Q_s$ be distinct rational places of $F$, and $Q_0$ be another place of $F$. Assume that $P_{0,1},\ldots,P_{0,m}$ are all places of $F'$ lying over $Q_0$.Let $G=nQ_{0}$, $D=Q_1+\cdots+Q_s$, and $G_j$ are $n$ divisor of $F'$ with
\[
	\text{supp}(G_j)\subset\{P_{0,1},\ldots,P_{0,m}\}.
\]
Define the evaluation map:
\begin{align*}
	\text{ev}_{D}: \mathcal{L}(G_1)\oplus\cdots\oplus\mathcal{L}(G_n)&\to\Pi=\bigoplus_{i=1}^s\mathbb{F}_q^{n\times n}\\
		(f_1,\ldots,f_n) &\mapsto (\bar{\varepsilon}_{Q_1}(L),\ldots,\bar{\varepsilon}_{Q_s}(L))
\end{align*}
where $L=\sum\limits_{i=1}^nf_i\sigma_i$. Since $f_1,\ldots,f_n$ have poles only at $P_{0}$, they belong to $\mathcal{O}'_{Q_i}$ for all $1\le i\le s$. Thus $\varepsilon_{Q_i}(L)$ are well-defined. Note that $\mathcal{L}(G_j)$ is an $\mathbb{F}_q$-linear space. Since $L$ is $F$-linear, it follows that $\text{ev}_{D}$ is an $\mathbb{F}_q$-linear map. Thus the image of $\text{ev}_{D}$ is a subspace of $\Pi$. The code $C_{sr}(D,G)$ is defined as the image of $\text{ev}_{D}$, which is exactly a linear space over $\mathbb{F}_q$.

\begin{remark}
In the work of \cite{Berardini}, E. Berardini and X. Caruso constructed codes through the isomorphism chain
\[
	D_{F'_i,x}\to D_{F'_i,1}\to\text{End}_{F_i}(F'_i)
\]
where $F_i$ denotes the completion of $F$ at $Q_i$, and $F'_i$ denotes the product of completions of $F'$ at places lying over $Q_i$ for each $i=1,\ldots,s$. The functor for generating the matrix in \cite{Berardini} has form $\sum\limits_{i=0}^{n-1}f_i\Phi^i$, where $\Phi^r=x$ for some $0\neq x\in F$, and $f_i\in\Lambda_{L,x}(G')$. To determine the code parameters, E. Berardini and X. Caruso hypothesized that $D_{L,x}$ has no nonzero zero divisor.

Our approach directly employs the operator $L=\sum\limits_{i=1}^nf_i\sigma_i$, corresponding to elements of $D_{F',1}$ since $F'/F$ is Galois. However, the ring $D_{F',1}$ is not a domain because the polynomial $T^r-1\in F'[T]$ is reducible. For the determination of code parameters, we consider $Q_0$ that is totally ramified or splits completely in $F'/F$, and then present another choice for $f_i$, which is explicitly described in subsequent sections.
\end{remark}
\begin{proposition}
The length $n^2s$ of $C_{sr,1}$ over $\mathbb{F}_q$ satisfies 
\[
	n^2s\le n^2(q+2g(F)\sqrt{q}+1).
\]
\end{proposition}
\begin{proof}
In the construction of $C_{sr}$, each place $Q\neq Q_{0}$ yields a well-defined $n\times n$ matrix $\bar{\varepsilon}_P(L)$ over $\mathbb{F}_q$. By the Hasse-Weil bound, the function field $F/\mathbb{F}_q$ has at most $q+1+2g(F)\sqrt{q}$ rational places, which implies that $s\le q+2g(F)\sqrt{q}+1$.
\end{proof}

\section{Sum-Rank Metric Codes via Quadratic Extension of Algebraic Function Fields}
\subsection{Construction with Totally Ramified Rational Places}
We first establish some necessary lemmas before deriving the dimension and minimum sum-rank distance.
\begin{lemma}\label{nonzero1}
Let $P$ be a rational place of $F$. Suppose that $Q\neq P$ is a place of $F$ such that $Q$ is totally ramified in $F'/F$. Let $Q'$ be the place of $F'$ lying over $Q$. For $f_1,\ldots,f_n\in\mathcal{L}(kQ')$, we have
\[
	\det(\varepsilon_P(L))\in\mathcal{L}(kQ).
\]
Furthermore, if there exists an $f_i$ satisfying
\[
	v_{Q'}(f_i)<v_{Q'}(f_j),\forall j\neq i,
\]then $\det(\varepsilon_P(L))\neq0$.
\end{lemma}
\begin{proof}
Since $f_i\in\mathcal{L}(kQ')$, the coefficients of $B_{f_i}$ belong to $\mathcal{L}(kQ')$ for $i=1,\ldots,n$. Therefore, the discrete valuation of $\det(\varepsilon_P(L))$ at $Q'$ is greater than the minimal valuation in its determinant expansion. Consequently, we have $\det(\varepsilon_P(L))\in\mathcal{L}(knQ')$. As $Q$ is totally ramified, the conorm of $Q$ is
\[
	\text{Con}_{F'/F}(Q)=nQ'.
\]
By Lemma \ref{detmatrix} (ii) and the properties of the conorm, we obtain
\[
	v_{Q}(\det(\varepsilon_P(L)))\le k.
\]
Note that for any place $Q'$ of $F'$, there is a unique place $Q$ in $F$ such that $Q'|Q$, so $\det(\varepsilon_P(L))$ has poles only at $Q$. Hence, $\det(\varepsilon_P(L))\in\mathcal{L}(kQ)$. For any nonzero $g\in\mathcal{O}'_P$, the equality
\[
	v_{Q'}(\sigma_i(g))=v_{Q'}(\sigma_1(g))
\]
holds for all $i=1,\ldots,n$. W.l.o.g, assume that $v_{Q'}(f_1)<v_{Q'}(f_j)$ for all $j=2,\ldots,n$. This implies that
\[
	v_{Q'}(f_1\sigma_1(g))<v_{Q'}(f_j\sigma_j(g)).
\]
By the strict triangle inequality, it follows that
\[
	v_{Q'}\left(\sum\limits_{j=1}^nf_j\sigma_j(g)\right)=v_{Q'}(f_1\sigma_1(g))<v_{Q'}(\sigma_1(g)).
\]
Thus $L(g)\neq0$ for any $g\neq0$. The conclusion follows from Lemma \ref{detmatrix} (iii).
\end{proof}
\begin{lemma}\label{zeros}
Let $P$ be a place of $F$ with a prime element $t$. If $\det(\varepsilon_P(L))\neq0$ for $L$ associated to $Q\neq P$ while ${\rm rank}(\bar{\varepsilon}_{P}(L))<n$, then $\det(\varepsilon_{P}(L))\in P$.
\end{lemma}
\begin{proof}
The condition ${\rm rank}(\bar{\varepsilon}_{P}(L))<n$ means that $\det(\bar{\varepsilon}_{P}(L))=0$. Let $\pi_{P}$ denote the residue class map at $P$. For $\varepsilon_{P}(L)=(M_{ij})$, we have
\[
	\bar{\varepsilon}_{P}(L)=(\pi_{P}(M_{ij}))=(M_{ij})(P).
\]
Therefore
\[
	\det(\bar{\varepsilon}_{P}(L))=\det(\varepsilon_{P}(L))(P).
\]
Since $\det(\varepsilon_{P}(L))(P)=0$, it follows that $\det(\varepsilon_{P}(L))\in P$.
\end{proof}
Let $k,k_1$ be integers such that $k>k_1\ge2g(F)-1$, and $P$ be a rational place of $F$. Then by the Riemann-Roch theorem, we obtain $\mathcal{L}(kP)$ and $\mathcal{L}(k_1P)$ are vector spaces over $\mathbb{F}_q$ with dimensions
\[
	\ell_k:=k-g(F)+1, \ell_{k_1}:=k_1-g(F)+1,
\]
respectively. Let 
\[
	\mathcal{L}(kP)=\langle\omega_1,\omega_2,\ldots,\omega_{\ell_k}\rangle
\]
where $\{\omega_1,\omega_2,\ldots,\omega_{\ell_{k_1}}\}$ forms a basis for $\mathcal{L}(k_1P)$. We define
\[
	\mathcal{L}^*_{k_1}(kP):=\langle\omega_{\ell_{k_1}+1},\ldots,\omega_{\ell_k}\rangle.
\]
It is obvious that $\mathcal{L}^*_{k_1}(kP)$ is a vector space over $\mathbb{F}_q$ of dimension $k-k_1$. Then we have:
\begin{theorem}\label{firstcase}
Suppose that $F'/F$ is a Galois extension of degree 2 with a totally ramified rational place $Q_0$ in $F'/F$, and $Q'_0|Q_0$. Let $s,k$ and $k_1$ be integers satisfying $s\ge2$ and $2g(F)-1\le k_1<k<2s$. For distinct places $Q_1,\ldots,Q_s$, the code
\[
	C_{sr,1}:=\{(\bar{\varepsilon}_{Q_1}(L),\ldots,\bar{\varepsilon}_{Q_s}(L))|L=f_1\sigma_1+f_2\sigma_2,f_1\in\mathcal{L}^*_{k_1}(kQ'_0),f_2\in\mathcal{L}(k_1Q'_{0})\}
\]
has parameters $[4s,\ell_k,d_{sr}\ge 2s-k]$.
\end{theorem}
\begin{proof}
The length of $C_{sr,1}$ is exactly $4s$ over $\mathbb{F}_q$. Its dimension is determined by
\[
	k-k_1+\ell_{k_1}-\dim(\mathcal{L}(kQ_{\infty}-\text{Con}_{F'/F}(\sum\limits_{i=1}^sQ_s))).
\]
Since $k<2s$, the evaluation map $\text{ev}_{D,1}$ is injective, implying that the dimension of $C_{sr,1}$ reduces to $k-k_1+\ell_{k_1}=\ell_k$.

Let $f_1\in\mathcal{L}^*_{k_1}(kQ'_0)$ and $f_2\in\mathcal{L}(k_1Q'_{0})$ with $f_1\neq0$. Suppose that $rank(\bar{\varepsilon}_{Q_i}(L))=d_i$ for $i=1,\ldots,s$. Then the sum-rank weight of the codeword corresponding to $L$ is $d_L=\sum\limits_{i=1}^sd_i$. By Lemma \ref{zeros}, we have
\[
	v_{Q_i}(\det(\varepsilon_{Q_i}(L)))\ge 2-d_i
\]
for each $i$. Furthermore, Lemma \ref{nonzero1} shows that
\[
	\det(\varepsilon_{Q_i}(L))\in\mathcal{L}(kQ_0).
\]
From $v_{Q'_0}(f_1)<v_{Q'_0}(f_2)$ and $f_1\neq0$, it follows that $\det(\varepsilon_{Q_i}(L))\neq0$. 

Denote $F'=F(y)$. Since $[F':F]=2$, there exists a monic irreducible quadratic polynomial $\phi(T)\in F[T]$ that is the minimal polynomial of $y$ over $F$, write $\phi(T)=T^2+a_1T+a_2$.

When $f_1=0$, the matrix $\varepsilon_{Q_i}(L)$ takes the form
\begin{equation*}
A_{\sigma_2}\cdot\begin{bmatrix}
		&f_{21},&f_{22}\\
		&-f_{22}a_2,&f_{21}-f_{22}a_1
	\end{bmatrix}
\end{equation*}
where $f_2=f_{21}z_1+f_{22}z_2$. Then $\det(\varepsilon_{Q_i}(L))=0$ means 
\[
	f_{21}^2-f_{21}f_{22}a_1+f_{22}^2a_2=0.
\]
If $f_{22}=0$, the previous equation implies $f_{21}=0$ and $f_2=0$. If $f_{21}\neq0$, then the function $\frac{-f_{21}}{f_{22}}$ is a root of $\phi(T)$. Both cases lead to a contradiction. Thus $\det(\varepsilon_{Q_i}(L))$ cannot vanish, and it has at most $k$ zeros. This implies that
\[
	\sum\limits_{i=1}^s(2-d_i)\le k.
\]
Consequently, the codeword satisfies $d_L\ge2s-k$. Therefore, we have $d_{sr}\ge2s-k$.
\end{proof}

\subsection{Construction with Splitting Completely Rational Places}
For a Galois extension $F'/F$ of degree 2, the following lemmas are required to determine code parameters.
\begin{lemma}\label{nonzero2}
Let $P$ be a rational place of $F$. Suppose that $Q\neq P$ is a place of $F$ that splits completely in $F'/F$, and $Q_1,Q_2$ are the places of $F'$ lying over $Q$. For $f_i\in\mathcal{L}(kQ_i)$, we have
\[
	\det(\varepsilon_P(L))\in\mathcal{L}(kQ).
\]
Furthermore, if there exists an $f_i$ satisfying
\[
	v_{Q_i}(f_i)<v_{Q_j}(f_j),\forall j\neq i,
\]
then $\det(\varepsilon_P(L))\neq0$.
\end{lemma}
\begin{proof}
First, every term in the determinant expansion of $\varepsilon_P(L)$ is contained in $\mathcal{L}(n(Q_1+Q_2))$. Note that the conorm of $Q$ is
\[
	\text{Con}_{F'/F}(Q)=Q_1+Q_2.
\]
By an argument similar to that in Lemma \ref{nonzero1}, we obtain
\[
	\det(\varepsilon_P(L))\in\mathcal{L}(kQ).
\]

Second, w.l.o.g, assume that $v_{Q_1}(f_1)<v_{Q_2}(f_2)$. For any non-zero function $g$, we have
\[
	v_{Q_2}(\sigma_2(g))=v_{Q_1}(\sigma_1(g)).
\]
If $L(g)=0$, that is, $f_1\sigma_1(g)+f_2\sigma_2(g)=0$. Then by the strict triangle inequality and $v_{Q_1}(f_1)\le0$, it follows that
\[
	v_{Q_1}(\sigma_2(g))=v_{Q_1}(f_1\sigma_1(g))\le v_{Q_1}(\sigma_1(g)).
\]
This implies
\[
	v_{Q_2}(\sigma_1(g))\le v_{Q_2}(\sigma_2(g))
\]
and consequently
\[
	v_{Q_2}(f_1\sigma_1(g))<v_{Q_2}(f_2\sigma_2(g)).
\]
It follows that $v_{Q_2}(L(g))<\infty$, contradicting $L(g)=0$. Hence $L(g)\neq0$ for all non-zero $g\in\mathcal{O}'_P$. The conclusion follows from Lemma \ref{detmatrix} (iii).
\end{proof}
The subsequent result is analogous to Theorem \ref{firstcase}, so we omit proof details.
\begin{theorem}\label{secondcase}
Suppose that $F'/F$ is a function field extension of degree 2. Assume that there exists a place $Q_0$ of $F$ such that $Q_0$ splits completely in $F'/F$ with $Q_{0,1},Q_{0,2}$ lying over $Q_0$. Let $s,k$ and $k_1$ be integers satisfying $s\ge2$ and $2g(F)-1\le k_1<k<2s$. For distinct places $Q_1,\ldots,Q_s$, the code
\[
	C_{sr,2}:=\{(\bar{\varepsilon}_{Q_1}(L),\ldots,\bar{\varepsilon}_{Q_s}(L))|L=f_1\sigma_1+f_2\sigma_2,f_1\in\mathcal{L}^*_{k_1}(kQ_{0,1}),f_2\in\mathcal{L}(k_1Q_{0,2})\}
\]
has parameters $[4s,\ell_k,d_{sr}\ge 2s-k]$.
\end{theorem}

\subsection{Examples from Elliptic Function Fields}
An elliptic function field $\mathcal{E}/\mathbb{F}_q$ is a quadratic extension of the rational function field $\mathbb{F}_q(x)/\mathbb{F}_q$ and has genus 1. It serves as a canonical example and suffices to illustrate our framework.

In this section, we consider the case when $q$ is odd. Specifically, let $\mathcal{E}/\mathbb{F}_q=\mathbb{F}_q(x,y)$ where $x,y$ are transcendental over $\mathbb{F}_q$ and satisfy the smooth, irreducible equation:
\[
	y^2=f_{\mathcal{E}}
\]
where $f_{\mathcal{E}}\in\mathbb{F}_q[x]$ is a square-free cubic polynomial. The Galois group $\text{Aut}(\mathcal{E}/\mathbb{F}_q(x))=\{Id,\sigma\}$ is cyclic of order 2, where
\[
	\sigma(y)=-y.
\]
Let $P_{\infty}$ be the place at infinity of $\mathbb{F}_q(x)$, which is totally ramified in $\mathcal{E}/\mathbb{F}_q(x)$. Suppose that $Q_{\infty}$ is the place of $\mathcal{E}$ lying over $P_{\infty}$. The basis $\{1,y\}$ of $\mathcal{E}/\mathbb{F}_q(x)$ forms an integral basis for all places of $\mathcal{E}/\mathbb{F}_q$ except $P_{\infty}$. Consequently, for any rational place $P\neq P_{\infty}$, the matrix induced by $\sigma$ on $\mathcal{O}'_P$ over $\mathcal{O}_P$ is
\begin{equation}\nonumber
\begin{bmatrix}
1,&0\\
0,&-1
\end{bmatrix}.
\end{equation}
Let $f_1$ and $f_2$ be two functions over $\mathcal{E}/\mathbb{F}_q$, where each $f_i$ has the form $f_i=f_{i1}+f_{i2}y$ with $f_{i1}$ and $f_{i2}$ belonging to $\mathbb{F}_q(x)$ for $i=1,2$. For any rational place $P\neq P_{\infty}$, the matrix induced by $f_i$ on $\mathcal{O}'_P$ over $\mathcal{O}_P$ is
\begin{equation}\nonumber
\begin{bmatrix}
	f_{i1},&f_{i2}\\
	f_{i2}f_{\mathcal{E}},&f_{i1}
\end{bmatrix}.
\end{equation}
Let $L=f_1+f_2\sigma$. Then we obtain
\begin{equation}\nonumber
\varepsilon_{P}(L)=\begin{bmatrix}
		f_{11}+f_{21},&f_{12}+f_{22}\\
		(f_{12}-f_{22})f_{\mathcal{E}},&f_{11}-f_{21}
	\end{bmatrix}
\end{equation}
and
\[
	\det(\varepsilon_{P}(L))=f_{11}^2-f_{21}^2-(f_{12}^2-f_{22}^2)f_{\mathcal{E}}.
\]
For positive integers $k$ and $k_1$, define sum-rank codes $C_{sr,1}$ and $C_{sr,2}$ as follows:
\[
	C_{sr,1}:=\{(\bar{\varepsilon}_{Q_1}(L),\ldots,\bar{\varepsilon}_{Q_s}(L))|L=f_1+f_2\sigma,f_1\in\mathcal{L}^*_{k_1}(kQ_{\infty}),f_2\in\mathcal{L}(k_1Q_{\infty})\}
\]
and
\[
	C_{sr,2}:=\{(\bar{\varepsilon}_{Q_1}(L),\ldots,\bar{\varepsilon}_{Q_s}(L))|L=f_1+f_2\sigma,f_1\in\mathcal{L}^*_{k_1}(kQ_{0,1}),f_2\in\mathcal{L}(k_1Q_{0,2})\}
\]
where $Q_{0,1}$ and $Q_{0,2}$ are rational places lying over $Q_0$. By applying the Singleton bound together with Theorems \ref{firstcase} and \ref{secondcase}, we obtain
\[
	4s-2k\le 2d_{sr,1},2d_{sr,2}\le 4s-k+2,
\]
where $s\le q$. The following two examples are given by SageMath \cite{sage}.
\begin{example}
Let $q=7$ and consider the elliptic curve $\mathcal{E}: y^2=x^3+3$ defined over $\mathbb{F}_q$. The rational places of $\mathbb{F}_q(x)/\mathbb{F}_q$ are
\[
	\{P_{0}, P_{1},P_{2},P_{3},P_{4},P_{5},P_{6}\}
\]
where $P_{\alpha}$ corresponds to the prime element $x-\alpha$. For $k=6$ and $k_1=3$, the space $\mathcal{L}(6Q_{\infty})$ has basis $\{1,x,y,x^2,xy,x^3\}$. The code $C_{sr,1}(D,6Q_{\infty})$ has parameters $[28,6,8]$. Taking $f_1=3$ and $f_2=y$, we obtain
\begin{equation}\nonumber
	\varepsilon_{P}(L)=\begin{bmatrix}
		3,&1\\
		-(x^3+3),&3
	\end{bmatrix}
\end{equation}
with
\[
\det(\varepsilon_{P}(L))=x^3+5.
\]
Since $x^3+5$ is irreducible over $\mathbb{F}_q[x]$, the function $\det(\varepsilon_{P}(L))$ has no rational zeros. Consequently, the codeword
\begin{equation}\nonumber
\left(\begin{bmatrix}
		3,&1\\
		-3,&3
	\end{bmatrix},\ldots,
	\begin{bmatrix}
	3,&1\\
	5,&3
	\end{bmatrix}\right)
\end{equation}
achieves the maximum sum-rank weight of $14=2s$. Alternatively, taking $f_1=1$ and $f_2=x^3$, we obtain
\begin{equation}\nonumber
	\varepsilon_{P}(L)=\begin{bmatrix}
		x^3+1,&0\\
		0,& 1-x^3
	\end{bmatrix}
\end{equation}
with
\[
\det(\varepsilon_{P}(L))=1-x^6.
\]
All places in $\{P_{1},P_{2},P_{3},P_{4},P_{5},P_{6}\}$ are zeros of $\det(\varepsilon_{P}(L))$. Thus, the codeword
\begin{equation}\nonumber
	\left(\begin{bmatrix}
		1,&0\\
		0,&1
	\end{bmatrix},\ldots,
	\begin{bmatrix}
		0,&0\\
		0,&2
	\end{bmatrix}\right)
\end{equation}
has a sum-rank weight that attains the lower bound.
\end{example}
\begin{example}
Consider the same finite field $\mathbb{F}_q$ and elliptic curve $\mathcal{E}$ as in the previous example. We choose the place $P_1$, and let $Q_{1,1}$ and $Q_{1,2}$ be the places of $\mathcal{E}/\mathbb{F}_q$ lying over $P_1$. We take the rational places 
\[
\{P_{0},P_{2},P_{3},P_{4},P_{5},P_{6}\}.
\]
Let $D'=P_{0}+P_{2}+P_{3}+P_{4}+P_{5}+P_{6}$. Then the code $C_{sr,2}(D,6P_1)$ has parameters $[24,6,6]$
with sum-rank weight distribution:
\begin{align*}
	&A_0=1,A_6=36,A_7=144,A_8=1542,A_9=7944,\\
	&A_{10}=26904,A_{11}=46959,A_{12}=34122
\end{align*}
where $A_i$ counts the number of codewords with sum-rank weight $i$.
\end{example}
\subsection{The Gilbert-Varshamov-like Bound}
In this section, we asymptotically compare our codes with the GV-like bound for sum-rank metric codes. Suppose that $\{F_i\}_{i\in\mathbb{N}}$ is a tower of function fields such that $F_{i+1}/F_i$ is a Galois extension of degree 2 for all $i\in\mathbb{N}$. Denote by $\{g_i\}_{i\in\mathbb{N}}$ and $\{N(F_i)\}_{i\in\mathbb{N}}$ the genus and the number of rational places of $F_i$, respectively. Following the construction from previous sections, with a sequence of integers $\{k_i\}_{i\in\mathbb{N}}$ such that $2g_i-1\le k_i\ge 2N(F_i)$, we obtain a sequence of sum-rank metric codes $\{C_i\}_{i\in\mathbb{N}}$ with parameters
\begin{itemize}
\item code length: $s_i\le4N(F_i)$,
\item dimension: $\ell_{k_i}=k_i-g_i+1$,
\item minimal sum-rank distance: $d_{sr,i}\ge 2s_i-k_i$.
\end{itemize}
Therefore, we have
\[
	d_{sr,i}+\ell_{k_i}\ge 2s_i-g_i+1.
\]
By the Hasse-Weil lower bound over $\mathbb{F}_q$, if either $N(F_i)$ or $g_i$ tends to infinity, then so does the other. 

Ihara \cite{Ihara} observed that the Hasse-Weil bound is no longer optimal when the genus of a function field is sufficiently large with respect to $q$. Let $A(q)$ be the Ihara's constant, defined as
\[
	A(q) := \limsup\limits_{g\to\infty}\frac{\max\{N(F):\text{$F$ is a function field of genus $g$ over $\mathbb{F}_q$}\}}{g}.
\]
The Drinfeld-Vl{\v a}du{\c t} bound \cite{Vladut} states that $A(q)\le \sqrt{q}-1$. If $q$ is a square, then Drinfeld and Vl{\v a}du{\c t} proved that equality holds. Tsfasman, Vl{\v a}du{\c t}, and Zink \cite{TVZ-bound} provided an example of curves that attain this bound employing modular curves. Garcia and Stichtenoth \cite{Garcia} presented a more explicit tower of function fields. Using their technology, we can present a chain of quadratic function field extensions that attain the bound for some even $q$.
\begin{proposition}\label{tower}
Let $q=2^{2^h}$ with $h\ge2$, and $F_0=\mathbb{F}_q(z_0)$ be the rational function field. Suppose that $t=2^{h-1}$ and $\{F_i\}_{i\in\mathbb{N}}$ is a tower of function fields such that $F_{i+1}=F_i(z_{i+1})$, where $z_i$ satisfies the equations:
\begin{itemize}
\item for $j\ge0$, $z_{jt+1}^2+z_{jt+1}=x_{j}^{2^t+1}$, where $x_0=z_0$, $x_1=z_t$ and $x_j=\frac{z_{jt}}{x_{j-1}}$ for $j\ge2$,
\item for $j\ge0$ and $2\le i\le t$, $z_{jt+i}^2+z_{jt+i}=z_{jt+i-1}$.
\end{itemize}
Then
\[
	\lim\limits_{i\to\infty}\frac{N_i}{g_i}=\sqrt{q}-1.
\]
\end{proposition}
\begin{proof}
It is sufficient to show that for any $j\ge0$, the function field $F_{jt}$ has the form $F_{(j+1)t}=F_{jt}(z_{(j+1)t})$, where $z_{(j+1)t}$ satisfies the equation
\[
	z_{(j+1)t}^{2^t}+z_{(j+1)t}=x_{j}^{2^t+1},
\]
with
\[
	x_0=z_0,\ x_1=z_t,\ x_j=\frac{z_{jt}}{x_{j-1}}\text{ for }j\ge2.
\]
Then, using \cite[Corollary 3.2]{Garcia}, we obtain
\[
	\lim\limits_{i\to\infty}\frac{N_i}{g_i}=\lim\limits_{j\to\infty}\frac{N_{jt}}{g_{jt}}=\sqrt{q}-1.
\]

For $j=0$, observe that
\begin{align*}
&z_t^2+z_t=z_{t-1},\\
&z_t^4+z_t^2=z_{t-1}^2,\\
&\cdots\\
&z_t^{2^t}+z_t^{2^{t-1}}=z_{t-1}^{2^{t-1}}.
\end{align*}
Recall that $t=2^{h-1}$, which implies
\begin{align*}
	z_t^{2^t}+z_t&=(z_{t-1}^{2^{t-1}}+z_{t-1}^{2^{t-2}})+\cdots+(z_{t-1}^{2}+z_{t-1})\\
	&=z_{t-2}^{2^{t-2}}+z_{t-2}^{2^{t-4}}+\cdots+z_{t-2}^4+z_{t-2}\\
	&=(z_{t-2}^{2^{t-2}}+z_{t-2}^{2^{t-3}})+(z_{t-2}^{2^{t-3}}+z_{t-2}^{2^{t-4}})+\cdots+(z_{t-2}^4+z_{t-2}^2)+(z_{t-2}^2+z_{t-2})\\
	&=z_{t-3}^{2^{t-3}}+z_{t-3}^{2^{t-4}}+\cdots+z_{t-3}^2+z_{t-3}\\
	&=z_{t-4}^{2^{t-4}}+z_{t-4}^{2^{t-6}}+\cdots+z_{t-4}^4+z_{t-4}\\
	&=z_{t-5}^{2^{t-5}}+z_{t-5}^{2^{t-6}}+\cdots+z_{t-5}^2+z_{t-5}\\
	&\qquad \vdots\\
	&=z_1^2+z_1\\
	&=x_0^{2^t+1}
\end{align*}
Then we have $F_{t}=F_{0}(z_{t})$ with 
\[
	z_t^{2^t}+z_t=x_0^{2^t+1}
\]
The proof can be completed with similar equations for any $j\ge1$.
\end{proof}

The above proposition shows that there exist towers of quadratic function field extensions over $\mathbb{F}_q$ attaining the Drinfeld-Vl{\v a}du{\c t} bound for some square $q$. For any sum-rank metric code $C_i\in\mathbb{F}_q^{2\times2}$ with $i\in \mathbb{N}$, its rate $\mathcal{R}_{i}$ and relative minimum distance $\delta_{sr,i}$ are defined as
\[
	\mathcal{R}_{i}:=\frac{\ell_{k_i}}{2s_i},\quad \delta_{sr,i}:=\frac{d_{sr,i}}{2s_i}.
\]
\begin{theorem}
Suppose that $q$ is a square and there exists a tower $\{F_i\}_{i\in\mathbb{N}}$ of quadratic function field extensions over $\mathbb{F}_q$ attaining the Drinfeld-Vl{\v a}du{\c t} bound. For any real numbers $\mathcal{R}_{sr},\delta_{sr}\in(0,1)$ satisfying
\begin{align}\label{agbound}
	\mathcal{R}_{sr}<1-\delta_{sr}-\frac{1}{2(\sqrt{q}-1)}.
\end{align}
There exists a sum-rank metric code with rate at least $\mathcal{R}_{sr}$ and relative minimum distance at least $\delta_{sr}$.
\end{theorem}
\begin{proof}
We only need to consider the sequence of sum-rank metric codes $\{C_i\}_{i\in\mathbb{N}}$ constructed from the tower $\{F_i\}_{i\in\mathbb{N}}$ via Theorem \ref{firstcase} or \ref{secondcase}. For any $i\in\mathbb{N}$, the parameters of $C_i$ are 
\[
	[s\le4N(F_i),\ell_{k_i},d_{sr,i}\ge 2s_i-k_i].
\]
Denote the rate $\mathcal{R}_i$ and relative minimum distance $\delta_{sr,i}$ as above. Choose $k_i$ such that $\mathcal{R}_i\ge \mathcal{R}_{sr}$ for all $i$ and $\lim\limits_{i\to\infty}\mathcal{R}_i=\mathcal{R}_{sr}$. Then we have
\[
	\mathcal{R}_{sr}=\lim\limits_{i\to\infty}\mathcal{R}_i\ge1-\lim\limits_{i\to\infty}\delta_{sr,i}-\frac{1}{2(\sqrt{q}-1)}.
\]
Thus for sufficiently large $i$, the code $C_i$ satisfies the desired parameters.
\end{proof}

Note that the GV-like bound for $n=m=2$ has the form

\begin{align}\label{gvbound}
\mathcal{R}_{sr}<&(\delta_{sr}-1)^2-\frac{\delta_{sr}}{2}\log_q(1+\frac{1}{2\delta_{sr}})-\frac{\log_q(1+2\delta_{sr})}{4}\nonumber\\
				&-\frac{\log_q(\gamma_q)}{4}+o(1).
\end{align}
For certain values of $q$ such that there exists a tower of quadratic function field extensions attaining the Drinfeld-Vl{\v a}du{\c t} bound, some families of our sum-rank metric codes surpass the GV-like bound. Using the function field tower given in Proposition \ref{tower}, we plot the comparisons between equations \eqref{agbound} with \eqref{gvbound} for $q=2^4$, $q=2^8$ and $q=2^{16}$ in Figs. \ref{16}, \ref{256}, and \ref{65536}, respectively.
\begin{figure}
	\centering
	\includegraphics[scale=0.4]{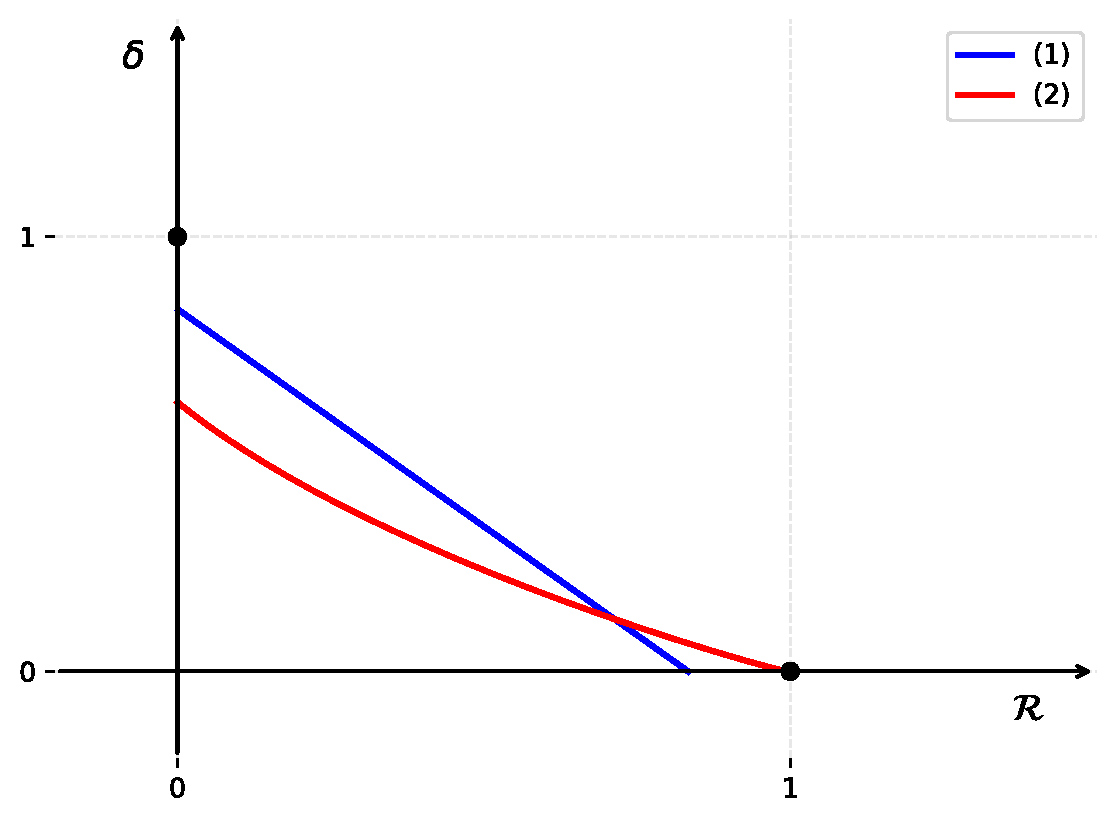}
	\caption{\text{$q=2^{4}$}}
	\label{16}
\end{figure}
\begin{figure}
	\centering
	\includegraphics[scale=0.4]{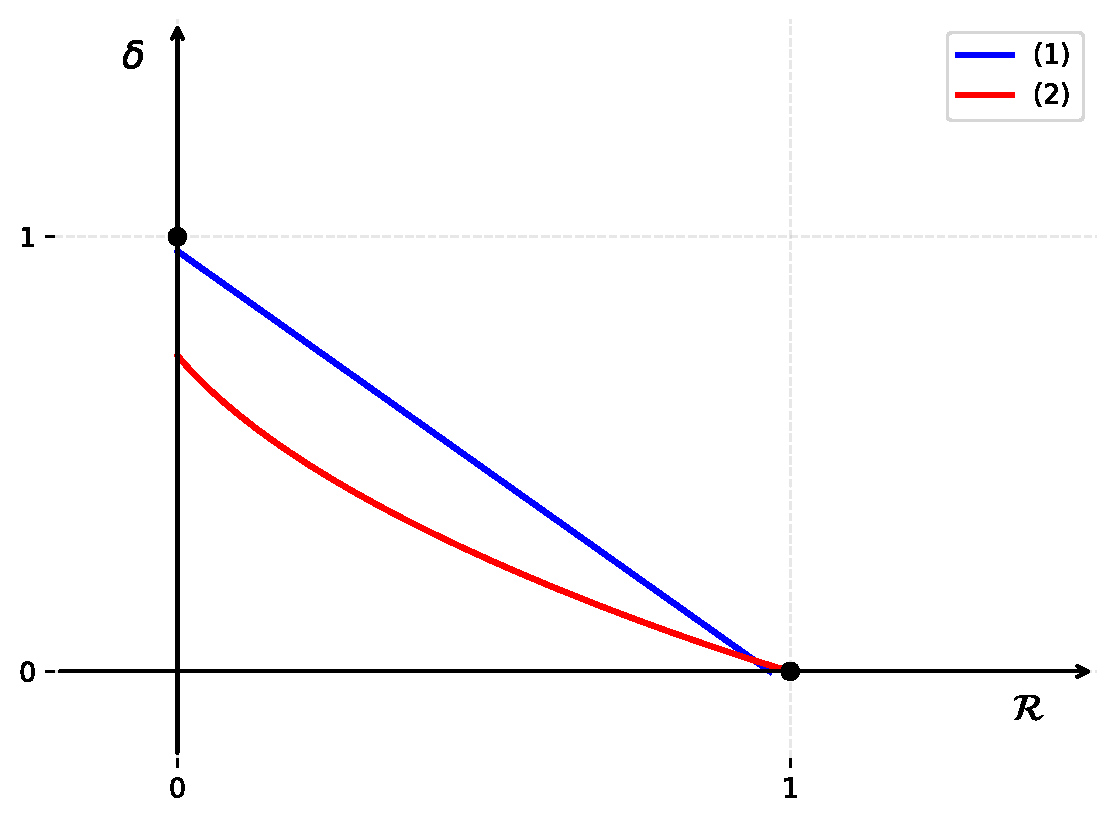}
	\caption{\text{$q=2^{8}$}}
	\label{256}
\end{figure}
\begin{figure}
	\centering
	\includegraphics[scale=0.4]{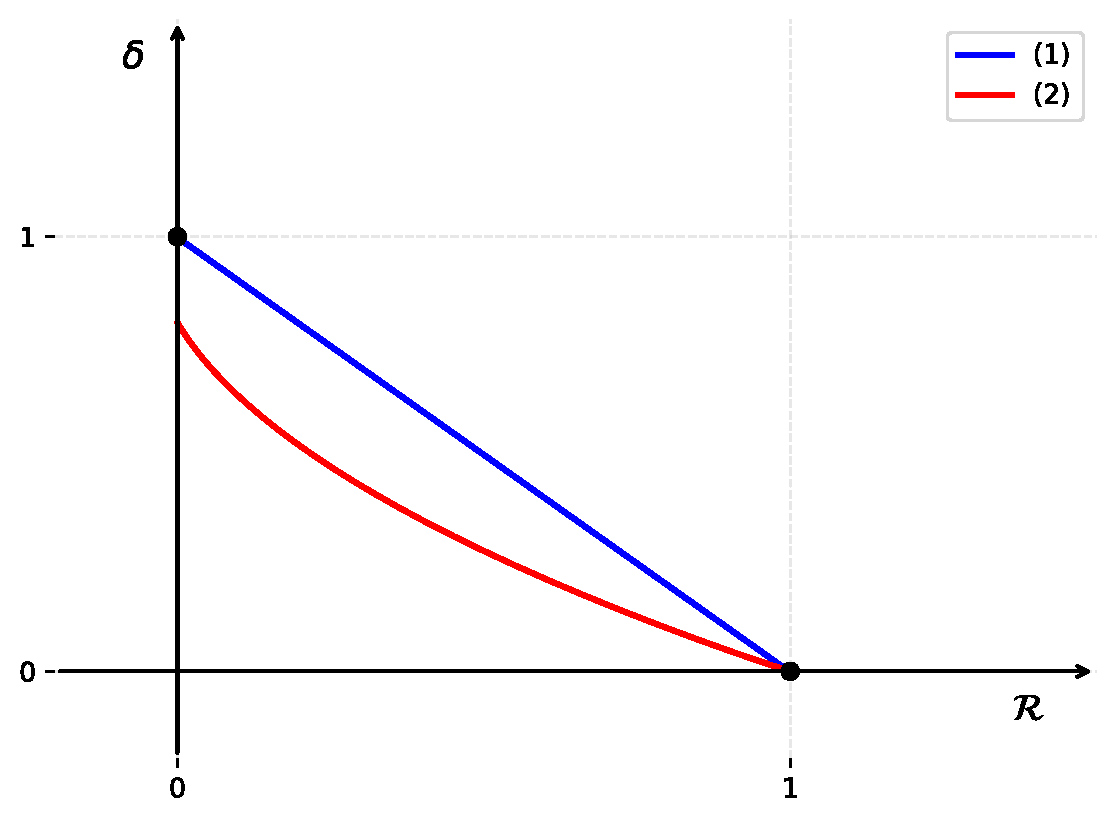}
	\caption{\text{$q=2^{16}$}}
	\label{65536}
\end{figure}

\section{Conclusion}
In this paper, we presented two constructions of $2\times2$ sum-rank metric codes via quadratic Galois extensions of algebraic function fields. By expressing the functions of extension fields to matrices over the base fields, we determined the parameters of these codes. Furthermore, we provided illustrative examples based on elliptic function fields.

To the best of our knowledge, these constructions are new. The only existing result on AG codes in the sum-rank metric was given by E. Berardini and X. Caruso \cite{Berardini}. In contrast to their work, our approach relies more heavily on the structures of function field extensions. In this paper, we established that $\det(\varepsilon_{P}(L))\neq0$ in the special case where $f_1=0$. The primary challenge in extending these constructions to $n\times n$ matrices lies in guaranteeing $\det(\varepsilon_{P}(L))\neq0$. Through SageMath computations, we verify that $L(g)$ has trivial kernel for certain special function fields and places when $n\ge3$. Consequently, in our future work, we will investigate more general constructions.

Regarding the sum-rank distance of our AG codes in the sum-rank metric, for a fixed dimension $k-g(F)+1$, we have
\[
	\max d_{sr}-\min d_{sr}\approx \frac{k+g(F)+1}{2}.
\]
This raises an interesting problem: adapting the techniques used in constructing MDS AG codes to the sum-rank metric to obtain MSRD AG codes. Building on the approach for AG codes with narrow bounds (such as elliptic codes \cite{mdselliptic}), we identify this problem as an important objective for our subsequent research.

\section{Acknowledgment}
This work is supported by Natural Science Foundation of Guangdong Province (No. 2025A1515011764), the National Natural Science Foundation of China (No. 12441107) and Guangdong Provincial Key Laboratory of Information Security Technology (No. 2023B1212060026).


\begin{thebibliography}{1}
\bibliographystyle{IEEEtran}
\bibitem{Abiad} A. Abiad, A. P. Khramova, and A. Ravagnani, ``Eigenvalue Bounds for Sum-Rank-Metric Codes," {\it IEEE Trans. Inf. Theory}, vol. 70, no. 7, pp. 4843-4855, July 2024.

\bibitem{Abida-bound} A. Abiad, A. L. Gavrilyuk, A. P. Khramova, and I. Ponomarenko, ``A Linear Programming Bound for Sum-Rank Metric Codes," in IEEE Transactions on Information Theory, vol. 71, no. 1, pp. 317-329, Jan. 2025.

\bibitem{Bartz} H. Bartz, T. Jerkovits, S. Puchinger and J. Rosenkilde, ``Fast Decoding of Codes in the Rank, Subspace, and Sum-Rank Metric," {\it IEEE Trans. Inf. Theory}, vol. 67, no. 8, pp. 5026-5050, Aug. 2021.

\bibitem{Berardini} E. Berardini and X. Caruso, ``Algebraic Geometry Codes in the Sum-Rank Metric," {\it IEEE Trans. Inf. Theory}, vol. 70, no. 5, pp. 3345-3356, May 2024.

\bibitem{Byrne} E. Byrne, H. Gluesing-Luerssen, and A. Ravagnani, ``Fundamental properties of sum-rank-metric codes," {\it IEEE Trans. Inf. Theory}, vol. 67, no. 10, pp. 6456-6475, Oct. 2021.

\bibitem{Byrne-anticode} E. Byrne, H. Gluesing-Luerssen, and A. Ravagnani, ``Anticodes in the sum-rank metric," {\it Linear Algebra Appl.}, vol. 643, pp. 80-98, Jun. 2022.

\bibitem{Moreno} E. Camps-Moreno, E. Gorla, C. Landolina, E. L. Garcia, U. Mart{\' i}nez-Pe{\~ n}as, and F. Salizzoni, ``Optimal anticodes, MSRD codes, and generalized weights in the sum-rank metric," {\it IEEE Trans. Inf. Theory}, vol. 68, no. 6, pp. 3806-3822, Jun. 2022.

\bibitem{Caruso} X. Caruso, A. Durand, ``Duals of linearized Reed-Solomon codes," {\it Des. Codes Cryptogr.}, vol. 91, pp. 241-271, Sep. 2023.

\bibitem{Chen-explicit} H. Chen, ``New Explicit Good Linear Sum-Rank-Metric Codes," {\it IEEE Trans. Inf. Theory}, vol. 69, no. 10, pp. 6303-6313, Oct. 2023.

\bibitem{Chen-decoding} H. Chen, Y. Qi, and Z. Cheng, ``Construction and Fast Decoding of Binary Linear Sum-Rank-Metric Codes", 2023, {\it arXiv:2311.03619}

\bibitem{mdselliptic} H. Chen, ``Many non-Reed-Solomon type MDS codes from arbitrary genus algebraic curves," {\it IEEE Trans. Inf. Theory}, Dec. 2023.

\bibitem{Delsarte} P. Delsarte, ``Bilinear forms over a finite field, with applications to coding theory," {\it J. Combinat. Theory A}, vol. 25, no. 3, pp. 226-241, Nov. 1978.

\bibitem{Gabidulin} E. M. Gabidulin, ``Theory of codes with maximum rank distance," {\it Problems Inf. Transmiss.}, vol. 21, no. 1, pp. 1-21, 1985.

\bibitem{Garcia} A. Garcia, H. Stichtenoth, ``A tower of Artin-Schreier extensions of function fields attaining the Drinfeld-S. G. Vl{\v a}du{\c t} bound," {\it Invent Math}, vol. 121, pp. 211-222, 1995. 

\bibitem{Guo} Z. Guo, C. Xing, C. Yuan, and Z. Zhang, ``Random Gabidulin Codes Achieve List Decoding Capacity in the Rank Metric," in {\it 2024 IEEE 65th Annual Symposium on Foundations of Computer Science (FOCS)}, Nov. 2024, pp. 1846-1873.

\bibitem{Goppa} V. D. Goppa, ``Algebraico-geometric codes," {\it Math. USSR-Izvestiya}, vol. 21, no. 1, pp. 75-91, 1983.

\bibitem{Horman} F. H{\" o}rmann, H. Bartz, and S. Puchinger, ``Error-Erasure Decoding of Linearized Reed-Solomon Codes in the Sum-Rank Metric," in {\it Proc. IEEE Int. Symp. Inf.}, Aug. 2022, pp. 7-12.

\bibitem{Horman-skew} F. H{\" o}rmann, and H. Bartz, ``Interpolation-based decoding of folded variants of linearized and skew Reed-Solomon codes", {\it Des., Codes Cryptogr.}, vol.92, no. 3, pp. 553-586, Mar. 2024.

\bibitem{Huffman} W. C. Huffman, and V. Pless, {\it Fundamentals of Error-Correcting Codes}, Cambridge University Press, Cambridge, 2003.

\bibitem{Ihara} Y. Ihara, ``Some remarks on the number of rational points of algebraic curves over finite fields," {\it J. Fac. Sci., Univ. Tokyo. Sect. A, Math.}, vol. 28, no. 3, pp. 721–724, 1982.

\bibitem{Islam} H. Islamm and A. L. Horlemann, ``Galois Hull Dimensions of Gabidulin Codes," in {\it 2023 IEEE Information Theory Workshop (ITW)}, Jun. 2023, pp. 42-46.

\bibitem{Jerkovits} T. Jerkovits, H. Bartz, and A. Wachter-Zeh, ``Randomized Decoding of Linearized Reed-Solomon Codes Beyond the Unique Decoding Radius," in {\it Proc. IEEE Int. Symp. Inf. Theory (ISIT)}, Aug. 2023, pp. 820-825.

\bibitem{Li} R. Li, and F. Fu, ``On the New Rank Metric Codes Related to Gabidulin Codes," in {\it 2024 IEEE Information Theory Workshop (ITW)}, Dec. 2024, pp. 597-602.

\bibitem{Liu-multishot} H. Liu, H. Wei, A. Wachter-Zeh, and M. Schwartz, ``Linearized Reed-Solomon Codes With Support-Constrained Generator Matrix and Applications in Multi-Source Network Coding," {\it IEEE Trans. Inf. Theory}, vol. 71, no. 2, pp. 895-913, Feb. 2025.

\bibitem{Penas-skew} U. Mart{\' i}nez-Pe{\~ n}as, ``Skew and linearized Reed-Solomon codes and maximum sum rank distance codes over any division ring," {\it J. Algebra}, vol. 504, pp. 587-612, Jun. 2018.

\bibitem{Penas-multishot} U. Mart{\' i}nez-Pe{\~ n}as, and F. R. Kschischang, ``Reliable and secure multishot network coding using linearized Reed-Solomon codes," {\it IEEE Trans. Inf. Theory}, vol. 65, no. 8, pp. 4785-4803, Aug. 2019.

\bibitem{Penas-mr} U. Mart{\' i}nez-Pe{\~ n}as, and F. R. Kschischang, ``Universal and dynamic locally repairable codes with maximally recoverablity via sum-rank codes," {\it IEEE Trans. Inf. Theory}, vol. 65, no. 12, pp. 7790-7805, Dec. 2019.

\bibitem{Penas-bch} U. Mart{\' i}nez-Pe{\~ n}as, ``Sum-rank BCH codes and cyclic-skew-cyclic codes," {\it IEEE Trans. Inf. Theory}, vol. 67, no. 8, pp. 5149-5167, Aug. 2021.

\bibitem{Penas-book} U. Mart{\' i}nez-Pe{\~ n}as, M. Shehadeh, and F. R. Kschischang, ``Codes in the sum-rank metric: Fundamentals and applications," {\it Found. Trends Commun. Inf. Theory}, vol. 19, no. 5, pp. 814-1031, 2022.

\bibitem{Penas-pmds} U. Mart{\' i}nez-Pe{\~ n}as, ``A general family of MSRD codes and PMDS codes with smaller field sizes from extended Moore matrices," {\it SIAM J. Discrete Math.}, vol. 36, no. 3, pp. 1868-1886, Sep. 2022.

\bibitem{Penas-dtextend} U. Mart{\' i}nez-Pe{\~ n}as, ``Doubly and triply extended MSRD codes," {\it Finite Field Appl.} vol. 91, 2023.

\bibitem{Penas-msrd} U. Mart{\' i}nez-Pe{\~ n}as, ``New constructions of MSRD codes," {\it Comp. Appl. Math}, vol. 43, no. 398, Sep. 2024.

\bibitem{Napp-multishot} D. Napp, R. Pinto, and V. Sidorenko, ``Concatenation of convolutional codes and rank metric codes for multi-shot network coding," {\it Des., Codes Cryptogr.}, vol. 86, no. 2, pp. 303-318, Feb. 2018.

\bibitem{Neri-Gabidulin} A. Neri, A.-L. Horlemann-Trautmann, T. Randrianarisoa, and J. Rosenthal, ``On the genericity of maiximum rank distance and Gabidulin codes,” {\it Des., Codes Cryptogr.}, vol. 86, no. 2, pp. 319-340, Apr. 2018.

\bibitem{Neri-twisted} A. Neri, ``Twisted linearized Reed-Solomon codes: A skew polynomial framework," {\it J. Algebra}, vol. 609, pp. 792-839, Nov. 2022.

\bibitem{Neri-oneweight} A. Neri, P. Santonastaso and F. Zullo, ``The geometry of one-weight codes in the sum-rank metric," {\it J. Combinat. Theory A}, vol. 194, 2023.

\bibitem{Nobrega-multishot} R. W. Nobrega and B. F. Uchoa-Filho, ``Multishot codes for network coding using rank-metric codes," in {\it Proc. 3rd IEEE Int. Workshop Wireless Netw. Coding}, Jun. 2010, pp. 1-14.

\bibitem{Ore} O. Ore, ``Theory of non-commutative polynomials," {\it Ann. Math.}, vol. 34, no. 3, p. 480, Jul. 1933.

\bibitem{Ott-bound} C. Ott, S. Puchinger, and M. Bossert, ``Bounds and genericity of sumrank-metric codes," in {\it Proc. XVII Int. Symp. Problems Redundancy Inf. Control Systems (REDUNDANCY)}, Oct. 2021, pp. 119-124.

\bibitem{Punchinger-bound} S. Puchinger and J. Rosenkilde, ``Bounds on List Decoding of Linearized Reed-Solomon Codes," in {\it Proc. IEEE Int. Symp. Inf. Theory (ISIT)}, Sep. 2021, pp. 154-159.

\bibitem{Puchinger} S. Puchinger, J. Renner, and J. Rosenkilde, ``Generic Decoding in the Sum-Rank Metric," {\it IEEE Trans. Inf. Theory}, vol. 68, no. 8, pp. 5075-5097, Aug. 2022.

\bibitem{Shehadeh} M. Shehadeh and F. R. Kschischang, ``Rate-diversity optimal multiblock space-time codes via sum-rank codes," in {\it Proc. IEEE Int. Symp. Inf. Theory (ISIT)}, Jun. 2020, pp. 3055-3060.

\bibitem{Stich} H. Stichtenoth, {\it Algebraic Function Fields and Codes} (Graduate Texts in Mathematics), vol. 254. Berlin, Germany: Springer-Verlag, 2009.

\bibitem{sage} SageMath, the Sage Mathematics Software System (Version 9.8), Sage Developers, Newcastle Upon Tyne, U.K., 2023. [Online]. Available: https://www.sagemath.org

\bibitem{TVZ-bound} M. A. Tsfasman, S. G. Vl{\v a}du{\c t}, and T. Zink, ``On Goppa codes which are better than the Varshamov–Gilbert bound,” {\it Math. Nachr}, vol. 109, pp. 21–28, Jan. 1982.

\bibitem{Tsfasman} M. A. Tsfasman and S. G. Vl{\v a}du{\c t}, {\it Algebraic-Geometric Codes}, Norwell, MA, USA: Kluwer Academic, 1991.

\bibitem{Vladut} S. G. Vl{\v a}du{\c t} and V. G. Drinfel’d, ``Number of points of an algebraic curve," {\it Funct. Anal. Appl.}, vol. 17, no. 1, pp. 53–54, Jan. 1983.

\end{thebibliography}
\end{document}